\documentclass[runningheads]{llncs}
\usepackage[T1]{fontenc}
\usepackage{graphicx}
\usepackage[noend]{algpseudocode}
\usepackage{shuffle}

\usepackage[nounderscore]{syntax}
\usepackage{csquotes}

\usepackage{caption}
\usepackage{subcaption}

\usepackage{booktabs}
\usepackage{marvosym}

\usepackage{graphicx}
\usepackage{amsmath}

\usepackage{amsthm,amssymb,mathtools}
\usepackage{wasysym}

\newtheorem*{proposition*}{Proposition}
\newtheorem*{theorem*}{Theorem}
\newtheorem*{lemma*}{Lemma}

\theoremstyle{definition}

\theoremstyle{remark}

\usepackage{graphics}
\usepackage[dvipsnames]{xcolor}
\usepackage{xspace}
\usepackage{tikz}
\usepackage{utfsym}
\usepackage{color}
\usepackage{relsize}

\newcommand{\ster}[1]{{*_{\mathlarger{#1}}}}

\newcommand{\N}{\ensuremath{\mathbb{N}}}
\newcommand{\Lb}{\ensuremath{\mathfrak{L}}}
\newcommand{\Id}{\ensuremath{\mathfrak{I}}}
\newcommand{\Lang}{\ensuremath{\mathcal{L}}}

\newcommand{\Sym}{\ensuremath{\mathcal{S}}}

\newcommand{\C}{\ensuremath{\mathcal{C}}}

\newcommand{\di}{\ensuremath{\hat\imath}}

\newcommand{\conc}[1]{\left\Vert{\smash{#1}}\right\Vert}
\newcommand{\concr}{\conc{\nobreak\hspace{.06em}}}
\newcommand{\pconc}[1]{\text{\textnormal{\textbrokenbar\textbrokenbar}}\smash{#1}\text{\textnormal{\textbrokenbar\textbrokenbar}}}
\newcommand{\pconcr}{\pconc{\nobreak\hspace{.12em}}}
\newcommand{\size}[1]{\left\vert{#1}\right\vert}

\usepackage{amssymb}
\usepackage{scalerel, tikz}
\usetikzlibrary{arrows.meta, bending}
\usetikzlibrary{arrows}

\newsavebox{\veeto}
\newsavebox{\veefrom}

\sbox{\veeto}  {\tikz{\draw[-{triangle 60}, line width=2pt, line join=round, line cap=round] (0,1)to[out=-10, in=90](.5,0)to[out=90, in=190](1,1)}}
\sbox{\veefrom}{\tikz{\draw[{triangle 60}-, line width=2pt, line join=round, line cap=round] (0,1)to[out=-10, in=90](.5,0)to[out=90, in=190](1,1)}}

\newcommand{\vto}{\mathbin{\scalerel*{\usebox{\veeto}}{t}}}
\newcommand{\vfrom}{\mathbin{\scalerel*{\usebox{\veefrom}}{t}}}

\newcommand{\wildcard}[0]{\mbox{\ensuremath{\bullet}}}
\newcommand{\altern}[0]{\ensuremath{\curlyvee}}

\newcommand{\alternleft}{\vfrom}
\newcommand{\alternright}{\vto}

\makeatletter
\newcommand*{\closeindex}[1]{_{\mkern-4.5mu#1}}
\let\sigmareal\altern
\DeclareRobustCommand{\altern}{\sigmareal\@ifnextchar_{\expandafter\closeindex\@gobble}{}}
\makeatother

\newcommand{\pto}{\ensuremath{\dashrightarrow}}

\newcommand{\DAseq}[3]{\ensuremath{%
    \Sym%
    \if\relax\detokenize{#1}\relax\else{,\,#1}\fi%
    ~|~%
    \Gamma%
    \if\relax\detokenize{#2}\relax\else{,\,#2}\fi%
    \Rightarrow%
    \Delta%
    \if\relax\detokenize{#3}\relax\else{,\,#3}\fi%
}}

\allowdisplaybreaks

\usepackage{stmaryrd}

\makeatletter
\RequirePackage[bookmarks,unicode,colorlinks=true]{hyperref}%
   \def\@citecolor{blue}%
   \def\@urlcolor{blue}%
   \def\@linkcolor{blue}%

\def\orcidID#1{\href{http://orcid.org/#1}{\protect\raisebox{-1.25pt}{\protect\includegraphics{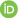}}}}
\makeatother

\begin{document}
%
\title{A Denotational Semantics for Synchronized Regular Expressions\thanks{Supported by the ATHENE project “ABC --- Control-Flow Fingerprinting for Malware Detection”}\\(extended version)}
\titlerunning{A Denotational Semantics for Synchronized Regular Expressions}
%
\author{Lukas Grätz\orcidID{0000-0002-9716-3142}}
%
%
\institute{Technical University of Darmstadt, Darmstadt, Germany\\
\email{lukas.graetz@tu-darmstadt.de}}
\maketitle              
\begin{abstract}
\emph{Pure} and \emph{full synchronized regular expressions} (\emph{pure sregex} and \emph{full sregex}) extend regular expressions by adding labels to the operators (Kleene star and alternation). Operators within the same scope are synchronized if they have the same label.
We show that a regular expression with backreferences (known from practical regex engines) can be translated into a full sregex.
Previous work did not consider synchronized alternations or nested star synchronization. Within the same scope, stars with the same label match the same number of iterations, alternations with the same label match the same choice. The scope may depend on the current iteration of an outer star.
Pure sregexes restrict synchronization to the local scope, while full sregexes also allow synchronization across iterations.
In contrast to operational semantics in previous work on backreferences, we present a denotational semantics, where \emph{concretizations} act similarly to valuations in logic.
As with backreferences, we show that matching a word is NP-complete.
\emph{Pure} and \emph{full synchronized languages} are considered and we show that they are closed under concatenation, union and star.
By a pumping lemma, they are not closed under intersection.
Relationships to other language models are discussed.

\end{abstract}

\section{Introduction}

Regular languages generated by \emph{regular expressions} \cite{RegexKleene} are the foundation of the Chomsky hierarchy for formal languages. As such, they are used by a \emph{lexer} within a compiler, like in~\cite{RegexLexer}. On the other hand, \emph{practical regexes} (or \emph{extended regexes}) are used for text processing, namely pattern matching and search-and-replace. Practical regex engines can be found in text editors and many programming languages' standard libraries.
They are also much more powerful than regular languages: for example, perlre supports \emph{backreferences}, \emph{(non-greedy) quantifiers}, \emph{lookaround assertions} and \emph{recursion} \cite{perlre}.
The trade-off is worse complexity: word membership testing for practical regex is often NP-hard \cite{Angluin,HandbookFindingPatterns,PerlNP,OnExtendedRegex}.
When worst-case run-time matters for security (e.g.\ ReDoS attacks), applications should avoid practical regex to allow running Thompson's DFA algorithm \cite{Thompson,ReDoS,RegexCanBeSimple,RegexLinguaFranca}.
Nevertheless, text processing applications seem to benefit from the increased expressivity, as practical regexes are popular.

In this paper, we explore an extension of regular expressions using \emph{synchronized} operators \cite{SRE}.
As with backreferences, word membership is NP-complete.
Our approach is to add labels $l\in\Lb$ to the operators, so \emph{alternations $\altern$} (in the literature $+$, $\cup$, $\lor$, or $|$) and \emph{Kleene stars $^*$} become $\altern_l$ and $^{\ster{l}}$.
Synchronization is done for two operators with the same label.
For example, $(\mathtt{a}\altern_{1}\mathtt{b})\mathtt{c}(\mathtt{d}\altern_1\mathtt{e})$ with the label $1$ matches only $\mathtt{acd}$ and $\mathtt{bce}$, while $(\mathtt{a}\altern_{1}\mathtt{b})\mathtt{c}(\mathtt{d}\altern_2\mathtt{e})$ with labels $1$ and $2$ corresponds to $(\mathtt{a}\altern\mathtt{b})\mathtt{c}(\mathtt{d}\altern\mathtt{e})$.
Previous work \cite{SRE} did not cover synchronized alternations $\altern_l$ or synchronized inner stars as in $(\mathtt{a}^\ster{1}\mathtt{bc}^\ster{1})^\ster{2}$.

There are some features even context-free grammars cannot generate:
In the \emph{Handbook of Formal Languages} \cite{Dassow1997}, three basic features of natural languages are listed; also compare~\cite{CharacterisingRegEx}.
All three can be modeled with synchronized regular expressions, where the \emph{wildcard} (a character class) is an abbreviation $\wildcard_{l} \equiv \mathtt{a} \altern_{\mathtt{a},l} \mathtt{b} \altern_{\mathtt{b},l} \dotsc \mathtt{z} \altern_{\mathtt{z},l} \emptyset$ for $\Sigma = \{\mathtt{a},\mathtt{b},\dotsc, \mathtt{z}\}$ and a label $l$:
\begin{enumerate}
\item For \emph{reduplications} like $\{ww\mid w\in\Sigma^*\}$, we can use ${\wildcard_{1}}^{\ster{2}}{\wildcard_{1}}^{\ster{2}}$.
\item For \emph{multiple agreements} like $\{\mathtt{a}^n \mathtt{b}^n \mathtt{c}^n\mid n\in\N\}$, we can use
$\mathtt{a}^{\ster{1}} \mathtt{b}^{\ster{1}} \mathtt{c}^{\ster{1}}$.
\item For \emph{crossed agreements} like $\{\mathtt{a}^n \mathtt{b}^m \mathtt{c}^n \mathtt{d}^m\mid n,m\in\N\}$, we can use
$\mathtt{a}^{\ster{1}} \mathtt{b}^{\ster{2}} \mathtt{c}^{\ster{1}} \mathtt{d}^{\ster{2}}$.
\end{enumerate}

An important aspect is the treatment of operators nested inside a star. The regular expression $(\mathtt{ab}^{*})^*$ generates the language $\{\mathtt{ab}^{m_1}\cdots\mathtt{ab}^{m_n}\mid n\in\N, m_1,\dotsc,m_n\in\N\}$ and the synchronized variant $(\mathtt{ab}^{\ster{1}})^{\ster{2}}$ should generate the same language. Hence, $^{\ster{1}}$ stands for a different number of iterations when in different iterations of $^{\ster{2}}$. Thus, each iteration forms a different context to evaluate the operators inside the subexpression. In fact, $L_1 = \{(\mathtt{ab}^{m})^n\mid n\in\N, m\in\N\}$ is not generated by a \emph{pure synchronized regular expression (pure sregex)}, which we show using a new \emph{pumping lemma}.
We also introduce \emph{full synchronized regular expressions (full sregexes)}:
by adding one arrow $\uparrow$ to a subexpression, the inner operators become relative to the next outer scope (with multiple arrows $\uparrow\cdots\uparrow$ we can go out of multiple scopes).
Thus, $(\mathtt{a}(\mathtt{b}^{\ster{1}})_\uparrow)^{\ster{2}}$ generates $L_1$.

The paper is structured as follows: Section~\ref{sec:defs} defines the syntax of pure sregexes. The semantics is in Section~\ref{sec:concretization}, where we introduce (partial) \emph{concretizations}. We get a pumping lemma for proving that a language is not pure synchronized.
In Section~\ref{sec:extended}, we introduce syntax and semantics of full sregex.
In Section~\ref{sec:symregCompl}, we show that matching a word $w$ with an expression $s$ is NP-complete.
In Section~\ref{sec:close}, we present closure properties of synchronized languages.
In Section~\ref{sec:expressive}, we compare pure/full synchronized languages with other classes such as context-free grammars and practical regexes with backreferences. The conclusion is in Section~\ref{sec:conc}.

\section{Syntax} \label{sec:defs}

First, we describe the set of \emph{labels}. The set of labels needs to be infinite to accommodate sregexes of arbitrary size. Often, we use integers $\N$ for the set of labels $\Lb$, but the definition leaves it open to use a different set:

\begin{definition}{}\label{def:symindex}
Given an alphabet $\Sigma$ and a set of labels $\Lb$, \emph{pure synchronized regular expressions (pure sregexes)} are inductively defined as all characters $s\in\Sigma$, the \emph{empty word~$\epsilon$}, the \emph{empty language~$\emptyset$}, 
\emph{concatenations $rs$},
labeled \emph{alternations $r \altern_l s$},
and labeled \emph{Kleene stars $r^{\ster{l}}$}
for all pure sregexes $r$ and $s$ and all labels $l\in\Lb$.
Alternation is left-associative, and the star has the strongest binding, followed by concatenation, followed by alternation.
\end{definition}

Note that concatenation $rs$, the empty word $\epsilon$ and the empty language $\emptyset$ come without a label: only alternations and Kleene stars have labels, as they have two or more options. The empty language $\emptyset$ is a regular language that is not recognized using regular expressions without $\emptyset$. Technically, the empty word $\epsilon$ could be defined as an abbreviation $\epsilon \equiv \emptyset^\ster{l}$. Because of left-associativity, we have ${\alpha_1\altern_{l_1}\alpha_2\altern_{l_2}\dotsc \altern_{l_{k-1}} \alpha_k} = (\dotsc(\alpha_1\altern_{l_1}\alpha_2)\altern_{l_2}\dotsc)\altern_{l_{k-1}} \alpha_k$.
Because of the binding strength, we have $ab^{\ster{l}} = a(b^{\ster{l}})$ and $\alpha \altern_1 \beta^{\ster{2}} = \alpha \altern_1 (\beta^{\ster{2}})$.

\section{Semantics}\label{sec:concretization}

Similar to Kripke structures in modal logic, we define the semantics of synchronized regular expressions with \emph{concretizations}. We get that a word $w$ is matched by $s$, iff there is some concretization $\C = (\Id, \hat{\imath}, \concr)$ such that $w = \conc{s}$.

\begin{definition}{} \label{def:concdef}
A \emph{concretization for pure sregex over labels $\Lb$} is a triple $\C = (\Id, \hat{\imath}, \concr)$, where a set $\Id$ is called \emph{worlds}, the designated $\hat{\imath} \in \Id$ is called \emph{initial world}, and for each operator $\altern_l$ and $^\ster{l}$ with label $l\in\Lb$
\begin{align*}
\conc{\altern_l}_{(\cdot)}&\colon \Id \to \{ \alternleft, \alternright \} &
\conc{^{\ster{l}}}_{(\cdot)} &\colon \Id \to \Id^*
\end{align*}
are functions. The functions respectively map worlds $i \in \Id$ to a choice $\conc{\altern_l}_{(i)} \in \{ \alternleft, \alternright \}$ or to a sequence of worlds $(j_1,\dotsc,j_n) = \conc{^{\ster{l}}}_{(i)}\in\Id^*$.
Choice \emph{left} $\alternleft$ is defined by $\alpha \alternleft \beta = \alpha$, while \emph{right} $\alternright$ is defined by $\alpha\alternright\beta = \beta$.
\end{definition}


\begin{definition}{} \label{def:conc}
Given a concretization $\C = (\Id, \hat{\imath}, \concr)$ and a pure sregex $s$, the concretization of $s$ in world $i\in\Id$, denoted as $\conc{s}_i$, is defined by:
\begin{align*}
\conc{sr}_i &= \conc{s}_i \conc{r}_i \\
\conc{s \altern_l r}_i &= 
   \conc{s}_i  \conc{\altern_l}_{(i)} \conc{r}_i
\\
\conc{s^{\ster{l}}}_i &= \conc{s}_{j_1}  \conc{s}_{j_2} \cdots  \conc{s}_{j_n} \quad \text{ with $(j_1,j_2,\dotsc,j_n) = \conc{^{\ster{l}}}_{(i)}$} \\
\conc{\mathtt{a}}_i &= \mathtt{a} \qquad \text{for any letter $\mathtt{a}\in\Sigma$}\\
\conc{\epsilon}_i &= \epsilon\\
\conc{\emptyset}_i &= \text{\textnormal{undef}}
\end{align*}
We define the \emph{concretization of $s$ with $\C$} as $\conc{s} = \conc{s}_{\hat{\imath}}$ in the initial world $\hat{\imath}$.
\end{definition}

Note that if $s$ uses the empty language $\emptyset$, the concretization $\conc{s}$ could be undefined. By syntactically excluding $\emptyset$, we can avoid undefined concretizations.

\begin{example}{}\label{ex:emptyundef}
Given the pure sregex $\mathtt{a} \altern_1 \emptyset$, two concretizations are notable. First, $\conc{^\ster{1}}_{(\hat{\imath})} = \alternleft$ results in $\conc{\mathtt{a} \altern_1 \emptyset} = (\mathtt{a} \alternleft \text{undef}) = \mathtt{a}$. On the other hand, ${\conc{^\ster{1}}_{(\hat{\imath})} = \alternright}$ results in an undefined concretization $\conc{\mathtt{a} \altern_1 \emptyset} = (\mathtt{a} \alternright \text{undef}) = \text{undef}$.
\end{example}

The important part of Definition~\ref{def:conc} is the concretization of $\conc{s^{\ster{l}}}_i$. Instead of just creating $n$ copies of $\conc{s}_i$, we pass a world from $(j_1,\dotsc,j_n) = \conc{^{\ster{l}}}_{(i)}$ to the concretization of each iteration. 
Now, the other operations are self-explanatory: for $\conc{sr}_{i}$, the world $i$ is just passed to the subexpressions and for $\conc{s\altern_l r}_i$
we evaluate the operator's function in world $i$.

\begin{example}{}\label{ex:alt_star}
For the pure sregex $(\mathtt{a}^{\ster{2}} \altern_{3} \mathtt{b})^{\ster{1}} c$ and a concretization $\C = \linebreak(\{i_0,i_1\},i_0,\concr)$ with $\conc{^{\ster{1}}}_{(i_0)} = (i_0,i_1,i_0)$, $\conc{^{\ster{2}}}_{(i_0)} = (i_0)$, $\conc{\altern_3}_{(i_0)} = \alternleft$, and $\conc{\altern_3}_{(i_1)} = \alternright$, we get the following word:
\begin{align*}
&\conc{(\mathtt{a}^{\ster{2}} \altern_{3} \mathtt{b})^{\ster{1}} \mathtt{c}}
 = \conc{(\mathtt{a}^{\ster{2}} \altern_{3} \mathtt{b})^{\ster{1}}}_{i_0} \conc{\mathtt{c}}_{i_0} \\
 =& \conc{\mathtt{a}^{\ster{2}} \altern_{3} \mathtt{b}}_{i_0} \conc{\mathtt{a}^{\ster{2}} \altern_{3} \mathtt{b}}_{i_1} \conc{\mathtt{a}^{\ster{2}} \altern_{3} \mathtt{b}}_{i_0} \conc{\mathtt{c}}_{i_0} \tag{since $\conc{^{\ster{1}}}_{(i_0)} = (i_0,i_1,i_0)$} \\
 =& \conc{\mathtt{a}^{\ster{2}}}_{i_0}\mathrel{\phantom{\altern_{1}}} \phantom{\mathtt{b}} \conc{\mathtt{a}^{\ster{2}} \altern_{3} \mathtt{b}}_{i_1} \conc{\mathtt{a}^{\ster{2}}}_{i_0}\mathrel{\phantom{\altern_{1}}} \phantom{b} \conc{\mathtt{c}}_{i_0}  \tag{since $\conc{\altern_3}_{(i_0)} = \alternleft$} \\
 =& \conc{\mathtt{a}^{\ster{2}}}_{i_0}\mathrel{\phantom{\altern_{3}}} \phantom{b} \conc{\mathtt{b}}_{i_1}\phantom{\mathtt{a}^{\ster{2}}}\mathrel{\phantom{\altern_{3}}} \conc{\mathtt{a}^{\ster{2}}}_{i_0}\mathrel{\phantom{\altern_{3}}} \phantom{b} \conc{\mathtt{c}}_{i_0} \tag{since $\conc{\altern_3}_{(i_1)} = \alternright$} \\
 =& {\conc{\mathtt{a}}_{i_0}}^{\phantom{^{\ster{2}}}}\mathrel{\phantom{\altern_{3}}} \phantom{b} \conc{\mathtt{b}}_{i_1}\phantom{\mathtt{a}^{\ster{2}}}\mathrel{\phantom{\altern_{3}}} {\conc{\mathtt{a}}_{i_0}}^{\phantom{^{\ster{2}}}}\mathrel{\phantom{\altern_{3}}} \phantom{b} \conc{\mathtt{c}}_{i_0}  = \mathtt{abac}  \tag{since $\conc{^{\ster{2}}}_{(i_0)} = (i_0)$}
\end{align*}
\end{example}

\subsection{Acyclic, Injective and Full Concretizations}

Note that in the example, the world $i_0$ was used multiple times for several subiterations, to make the description of the concretization shorter. However, when proving properties, some kind of uniqueness property will be helpful. We can use the following definition:

\begin{definition}\label{def:inj}
    For any concretization $\C = (\Id, \hat\imath, \concr{})$, we define a graph of the concretization by $G = (V,E)$ with vertices $V = \Id$ and edges:
    \begin{align*}
        E = \{ (i,j_1),\dotsc,(i,j_n) \mid l \in \Lb, i \in \Id, (j_1,\dotsc,j_n) = \conc{^\ster{l}}_{(i)} \}
    \end{align*}
    A concretization $\C$ is \emph{acyclic}, iff its graph is acyclic;
    \emph{injective}, iff its graph is injective;
    \emph{full}, iff its graph is surjective.
\end{definition}

In Example~\ref{ex:alt_star}, the concretization is not acyclic. For $(i_0,i_1,i_0) = \conc{^{\ster{1}}}_{(i_0)}$, there is a cycle back to the initial world $i_0$. For acyclic concretizations, we often need more worlds, which is always possible:

\begin{lemma}\label{lem:acyc}
    Given a concretization $\C = (\Id, \hat\imath, \concr)$, we can construct an acyclic concretization $\C' = (\Id', \hat\imath', \concr{}')$ with $\conc{s}_i = \conc{s}'_{i'}$ for all pure sregexes $s$ and all worlds $i\in\Id$.
    If $\Id$ is enumerable, the resulting $\Id'$ also is enumerable.
\end{lemma}
\begin{proof}
    The approach is to use tuples $\langle i,k\rangle$ for each world $i\in\Id$, where the additional $k\in\N$ stands for the depth of nested $^\ster{l}$ operators.
    More precisely, we define a new set of worlds $\Id' = \Id \times \N$, where each new world is a tuple of an original world $i \in \Id$ and a depth counter $k \in \N$. Then, the functions are:
    \begin{align*}
        \conc{\altern_{l}}'_{(\langle i,k\rangle )} &= \conc{\altern_l}_{(i)} \\
        \conc{^\ster{l}}'_{(\langle i,k\rangle)} &= (\langle j_1,k+1\rangle,\dotsc,\langle j_n,k+1\rangle) &\text{for } (j_1,\dotsc,j_n) = \conc{^\ster{l}}_{(i)}
    \end{align*}
    Because $k$ is always incremented, there can be no cycle in the graph.
    If the original concretization's graph has a cycle $i \to \cdots \to i$ for some world $i\in\Id$, the acyclic version $\langle i,0\rangle \to \cdots \to\langle i,k\rangle$ has no cycle.
    To get $\conc{s}_i = \conc{s}'_{i'}$, we define $i' = \langle i,0\rangle \in \Id'$ for each $i\in\Id$.
\end{proof}

To continue with Example~\ref{ex:alt_star}, Lemma~\ref{lem:acyc} gives us the concretization $\C' = (\{i_0,i_1\}\times \N, \langle i_0,0\rangle , \concr')$ with $\conc{^\ster{1}}_{(\langle i_0,0\rangle )}' = (\langle i_0,1\rangle,\langle i_1,1\rangle ,\langle i_0,1\rangle )$ and $\conc{^{\ster{2}}}_{(\langle i_0,1\rangle )}' = (\langle i_0,2\rangle)$. While $\C$ had a finite set of worlds, $\C'$ is now infinite. However, we will see later that acyclic \emph{partial} concretizations can be useful: with partial functions $\pconc{\altern_l}$ and $\pconc{^\ster{l}}$, we might only need a small, finite set of worlds.

The acyclic concretizations from Lemma~\ref{lem:acyc} are not injective or full.
We need this when considering full sregex in Section~\ref{sec:extended}.
Once we have an acyclic concretization, we can also achieve an injective, full concretization:

\begin{lemma}\label{lem:injfull}
    Given a concretization $\C = (\Id, \hat\imath, \concr)$ with enumerable worlds $\Id$, we can construct an acyclic, injective and full concretization $\C' = (\Id', \hat\imath', \concr{}')$ with $\conc{s}_i = \conc{s}'_{i}$ for all pure sregexes $s$ and all worlds $i\in\Id$.
\end{lemma}
\begin{proof}
The construction requires three steps. First, we get $\C$ acyclic with Lemma~\ref{lem:acyc}.

Second, we get $\C$ injective: if the the graph contains edges $(i_1,j)\in E$ and $(i_2,j)\in E$ for $i_1\neq i_2$, then this means, there are two functions $\conc{^\ster{l_1}}$ and $\conc{^\ster{l_2}}$ with both $(\dotsc, j, \dotsc) = \conc{^\ster{l_1}}_{(i_1)}$ and $(\dotsc, j, \dotsc) = \conc{^\ster{l_2}}_{(i_2)}$ (labels $l_1$ and $l_2$ may coincide).
The approach is to change the second to $(\dotsc, j', \dotsc) = \conc{^{\ster{l_1}}}_{(i_2)}$ where $j'$ is a new world behaving the same as $j$, so $\conc{\altern_{l}}_{(j')} = \conc{\altern_l}_{(j)}$ and $\conc{^\ster{l}}_{(j')} = \conc{^\ster{l}}_{(j)}$.
Because $\Id$ is enumerable, we can do that infinitely often until no such conflict $(i_1,j)\in E$ and $(i_2,j)\in E$ exists anymore and the concretization is injective.

Third, we get $\C$ full: if there is a world $i_0 \in \Id$ such that the graph does not contain any edge $(j,i_0)\in E$, then we add infinitely many new worlds $i_{-1},i_{-2},\dotsc$ with $\conc{\altern_l} = \alternleft$ and $\conc{^\ster{l}}_{(i_{-k-1})} = (i_{-k})$ for all $k\in\N$ and all labels $l\in\Lb$.
\end{proof}

\subsection{Pure Synchronized Languages}

The semantics of a classical regex $\alpha$ is the language $\Lang(\alpha)$ it generates. We can define $\Lang(\mathtt{a}) = \{\mathtt{a}\}$ for any $\mathtt{a}\in\Sigma$, $\Lang(\alpha\beta) = \{ w_1w_2 \mid w_1\in\Lang(\alpha), w_2\in\Lang(\beta) \}$, $\Lang(\alpha\altern\beta) = \Lang(\alpha)\cup \Lang(\beta)$, and $\Lang(\alpha^*) = \{ w_1\dotsc w_n \mid n\in\N, w_1,\dotsc,w_n\in\Lang(\alpha) \}$. For pure sregex, we define the generated language using concretizations instead:

\begin{definition}
    For a given pure sregex $s$ over $\Sigma,\Lb$, the generated language is:
    \begin{align*}
        \Lang(s) = \bigl\{ \conc{s} \bigm| \C =(\Id,\hat\imath, \concr{}) \text{ is a concretization} \bigr\}
    \end{align*}
    We say the language $L= \Lang(s)$ is \emph{pure synchronized}.
\end{definition}

This is well-defined by the following theorem:

\begin{theorem}\label{prop:adequacy}
    When $\alpha$ is a regular expression and pure sregex $\tilde{\alpha}$ has each operation $\altern$ and $^*$ with a unique label $l\in\Lb$, then $\Lang(\alpha) = \Lang(\tilde{\alpha})$.
\end{theorem}
\begin{proof}
    We show this by induction. In the base case, when $\alpha = \mathtt{a} \in \Sigma$, we get $\Lang(\alpha) = \{\mathtt{a}\}$ and $\Lang(\tilde{\alpha}) = \{ \mathtt{a} \}$.

    In the induction step, we have $\Lang(\alpha) = \Lang(\tilde{\alpha})$ and $\Lang(\beta) = \Lang(\tilde{\beta})$ and we get:
    \begin{align*}
        \Lang(\tilde{\alpha}\tilde{\beta})
        &= \bigl\{ \conc{\tilde{\alpha}\tilde{\beta}}_{\hat\imath} \bigm| \C = (\Id,\hat\imath, \concr{}) \bigr\} \\
        &= \bigl\{ \conc{\tilde{\alpha}}_{\hat\imath}\conc{\tilde{\beta}}_{\hat\imath} \bigm| \C = (\Id,\hat\imath, \concr{}) \bigr\} \\
        \intertext{Because the labels in $\tilde{\alpha}$ and $\tilde{\beta}$ are unique, the concretizations for $\tilde{\alpha}$ and $\tilde{\beta}$ are independent:}
        &= \bigl\{ \conc{\tilde{\alpha}}^1_{{\hat\imath}_1}\conc{\tilde{\beta}}^2_{\hat\imath_2} \bigm| \C_1 = (\Id_1,\hat\imath_1, \concr^1), C_2 = (\Id_2, \hat\imath_2, \concr^2) \bigr\} \\
        &= \bigl\{ w_1w_2 \bigm| w_1\in\Lang(\tilde{\alpha}), w_2\in\Lang(\tilde{\beta}) \bigr\} \\
        &= \bigl\{ w_1w_2 \bigm| w_1\in\Lang(\alpha), w_2\in\Lang(\beta) \bigr\} 
        = \Lang(\alpha\beta)
    \end{align*}
    For the alternation:
    \begin{align*}
        \Lang(\tilde{\alpha}\altern_l \tilde{\beta})
        &= \bigl\{ \conc{\tilde{\alpha}\altern_l\tilde{\beta}}_{\hat\imath} \bigm| \C = (\Id,\hat\imath, \concr{}) \bigr\} \\
        &= \bigl\{ \conc{\tilde{\alpha}}_{\hat\imath} \bigm| \C \text{ with } \conc{\altern_l}_{(\hat\imath)} = \alternleft \bigr\}
        \cup \bigl\{ \conc{\tilde{\beta}}_{\hat\imath} \bigm| \C \text{ with } \conc{\altern_l}_{(\hat\imath)} = \alternright \bigr\} \\
    \intertext{Because the label $l$ is unique, it does not appear in $\tilde{\alpha}$ and $\tilde{\beta}$:}
        &= \bigl\{ \conc{\tilde{\alpha}}_{\hat\imath} \bigm| \C = (\Id, \hat\imath, \concr{}) \bigr\} \cup \bigl\{ \conc{\tilde{\beta}}_{\hat\imath} \bigm| \C = (\Id,\hat\imath, \concr{}) \bigr\} \\
        &= \Lang(\tilde{\alpha}) \cup \Lang(\tilde{\beta}) = \Lang(\alpha) \cup \Lang(\beta) = \Lang(\alpha\altern\beta)
    \end{align*}
    And for the Kleene star:
    \begin{align*}
        \Lang(\tilde{\alpha}^{\ster{l}})
        &= \bigl\{ \conc{\tilde{\alpha}^{\ster{l}}}_{\hat\imath} \bigm| \C = (\Id,\hat\imath, \concr{}) \bigr\} \\
        &= \bigl\{ \conc{\tilde{\alpha}}_{j_1}\cdots \conc{\tilde{\alpha}}_{j_n} \bigm| \C = (\Id,\hat\imath, \concr{}) \text{ with } (j_1,\dotsc,j_n) = \conc{^{\ster{l}}}_{(\hat\imath)} \bigr\} \\
        &\subseteq \bigl\{ \conc{\tilde{\alpha}}_{\hat\imath_1}^1 \cdots \conc{\tilde{\alpha}}_{\hat\imath_n}^n \bigm| n\in\N, \C_1 = (\Id_1,\hat\imath_1, \concr^1),\dotsc,\C_n = (\Id_n,\hat\imath_n,\concr^n) \bigr\} \\
        &= \bigl\{ w_1 \cdots w_n \bigm| n\in\N, w_1,\dotsc,w_n \in \Lang(\tilde{\alpha}) \}
        =  \Lang(\alpha^*)
    \end{align*}
    The back-direction:
    \begin{align*}
        \Lang(\alpha^{*})
        &= \bigl\{ w_1 \cdots w_n \bigm| n\in\N, w_1,\dotsc,w_n \in \Lang(\tilde{\alpha}) \} \\
        &= \bigl\{ \conc{\tilde{\alpha}}_{\hat\imath_1}^1 \cdots \conc{\tilde{\alpha}}_{\hat\imath_n}^n \bigm| n\in\N, \C_1 = (\Id_1,\hat\imath_1, \concr^1),\dotsc,\C_n = (\Id_n,\hat\imath_n,\concr^n) \bigr\} \\
    \intertext{By using the coproduct $\oplus$ on the sets of indices $\Id_1,\dotsc, \Id_n$, we can unify the concretizations as one $\C = (\Id,\hat\imath, \concr{})$ with indices $\Id = \Id_1\oplus \dotsc\oplus \Id_n \oplus \{\hat\imath\}$:}
        &\subseteq \bigl\{ \conc{\tilde{\alpha}}_{i_1}\cdots \conc{\tilde{\alpha}}_{i_n} \bigm| n\in\N, \C=(\Id,\hat\imath, \concr{}), i_1,\dotsc,i_n \in \Id \bigr\} \\
        &= \bigl\{ \conc{\tilde{\alpha}}_{j_1}\cdots \conc{\tilde{\alpha}}_{j_n} \bigm| n\in\N, \C = (\Id, \hat\imath, \concr{}) \text{ with } (j_1,\dotsc,j_n) = \conc{^{\ster{l}}}_{(\hat\imath)} \bigr\} \\
        &= \Lang(\tilde{\alpha}^{\ster{l}}) \qedhere
    \end{align*}
\end{proof}

If we use the same label in multiple places, we can describe a non-regular language:

\begin{example}
    $\Lang(\mathtt{a}^{\ster{1}}\mathtt{b}^{\ster{1}}\mathtt{c}^{\ster{1}}) = \{ \mathtt{a}^n\mathtt{b}^n\mathtt{c}^n \mid n\in\N \}$, because with any concretization, we get the same iterations $(j_1,\dotsc,j_n) = \conc{^{\ster{1}}}_{(\hat\imath)}$ for all three occurrences of~$^{\ster{1}}$. This language is not regular and not context-free, as we can show with the respective pumping lemma.
\end{example}

There are situations when an operation with the same label is evaluated differently in the same concretization, because the concretization function also depends on the world. This happens if these operations appear in the scope of differently labeled Kleene stars. Or one inside, one outside a Kleene star.

\begin{corollary}\label{ex:starstar}
    $\Lang((\mathtt{a}^{\ster{1}})^{\ster{1}}) = \Lang((\mathtt{a}^{\ster{2}})^{\ster{1}}) = \Lang((\mathtt{a}^*)^*) = \Lang(\mathtt{a}^*)$. The same label~$1$ does not make a difference, because for any concretization, the outer $^{\ster{1}}$ is evaluated for world~$0$ and the inner $^{\ster{1}}$ is evaluated for some worlds $j_1,\dotsc,j_n$.
\end{corollary}

The pumping lemma is a powerful tool to show that a language is not regular. We generalize this lemma for pure synchronized languages:

\begin{theorem}[Pumping Lemma]\label{thm:pumping}
    For any pure synchronized language $L\subseteq \Sigma^*$, there are constant numbers $p,q\in\N$ such that for any word $w \in L$ with $|w| > p$, we can divide it into $w = x_0 y_1 x_1 \cdots y_q x_q$ such that the following holds:
    \begin{enumerate}
        \item $1 \le |y_1 \cdots y_q|$
        \item $|x_0| + |y_1\cdots y_q| \le p$
        \item $x_0 y_1^k x_1 \cdots y_q^k x_q \in L$ for all $k\in\N$
    \end{enumerate}
    For a language generated by a pure sregex $s$ with the number of occurrences of each label at most $q\in\N$, we can use the constants $q$ and
    $p = |s|$.
\end{theorem}
\begin{proof}
    If $L$ is pure synchronized, then there is some pure sregex $s$ with $L = \Lang(s)$. By Lemma~\ref{lem:acyc}, for any $w\in\Lang(s)$, there is some injective, acyclic concretization $\concr{}$ with $w = \conc{s}$. Suppose that $|w| > p = |s|$. Now, when computing the concretization of $s$, we can decompose $w = x_0 y_1 x_1 \dotsi y_n x_n$ such that the substrings $y_1,\dotsc,y_n$ originate from the evaluation of the concretization in the same world $j\in\Id$ as $y_1 = \conc{r_1}_j, \dotsc, y_n = \conc{r_n}_j$ with the following conditions:
    \begin{itemize}
        \item There is a labeled star $^\ster{l}$ and an iteration $i\in\Id$ such that $\conc{r_1}_j, \dotsc,\conc{r_n}_j$ are used to evaluate one iteration from $\conc{r_1^{\ster{l}}}_i, \dotsc,\conc{r_n^{\ster{l}}}_i$, respectively. 
        \item The subexpressions $r_1^{\ster{l}},\dotsc,r_n^{\ster{l}}$ are exactly the subexpressions with the labeled star~$^{\ster{l}}$ which are evaluated in the world $i$.

        Note that since $\concr{}$ is injective and acyclic, each $\conc{r_1}_j, \dotsc,\conc{r_n}_j$ is only used once in the evaluation of the concretization $w = \conc{s}$. Furthermore, since $\concr{}$ is injective, acyclic, and the subexpressions $r_1^{\ster{l}},\dotsc,r_n^{\ster{l}}$ are evaluated in the same world~$i$, these subexpressions must be distinct in~$s$.
        \item $1 \le |y_1\cdots y_n|$. Since $|w| > |s|$, there must be some subexpression $u = \conc{r^{\ster{l}}}_i$ for some labeled star~$^{\ster{l}}$ and some iteration $i\in\Id$ such that $1 \le |u|$. This means, there is at least one iteration $y = \conc{r}_j$ from $\conc{r^{\ster{l}}}_{i}$ with $1 \le |y|$. We can use this $y$ as one of the $y_1,\dotsc,y_n$ to satisfy $1 \le |y_1\cdots y_n|$.
        \item If one of the substrings $x_0$ or $y_1,\dotsc,y_n$ contains some proper substring $u$ which is evaluated by some $u = \conc{\tilde{r}^{\ster{\tilde{l}}}}_{\tilde{i}}$, then it is empty $|u| = 0$. Otherwise, instead of $y_1,\dotsc,y_n$, we could take $\tilde{y}_1 = \conc{\tilde{r}_1}_{\tilde{j}},\dotsc,\tilde{y}_{\tilde{n}} = \conc{\tilde{r}_{\tilde{n}}}_{\tilde{j}}$ for some iteration $\tilde{j}$ from $\conc{\tilde{r}^{\ster{\tilde{l}}}}_{\tilde{i}}$ to satisfy this condition.
    \end{itemize}

    Now, we can modify the concretization $\concr{}$ as $\concr^k$ with $k\in\N$, such that the iteration with world $j$ is repeated $k$ times in $\conc{^{\ster{l}}}^k_{(i)}$.
    Thus, $\concr^k$ is no longer injective.
    We get $\conc{s}^k = x_0 y_1^k x_1 \cdots y_n^k x_n$, where $x_0,x_1,\dotsc,x_n$ are the substrings before/after/between $y_1,\dotsc,y_n$ in $w$. Hence $x_0 y_1^k x_1 \cdots y_n^k x_n \in \Lang(s)$.

    The number $n$ of these substrings $y_1,\dotsc,y_n$ is at most the number of occurrences of $^{\ster{l}}$ with label $l$ in pure sregex $s$. This is bounded by $q$, the number of occurrences of the same label in $s$.
    Since $n \le q$, we can divide $w$ into $w = x_0 y_0 x_1 \cdots y_q x_q$ with $x_t = \epsilon$ and $y_t = \epsilon$ for $t = n{+}1, n{+}2, \dotsc, q$.
    The pumping lemma follows with $p = |s|$, since the subexpressions $x_0$ and $y_1,\dotsc,y_n$ are distinct in $s$:
    \begin{enumerate}
    \item $1 \le |y_1\cdots y_n| = |y_1\cdots y_n\epsilon \cdots \epsilon| = |y_1\cdots y_q|$
    \item $|x_0| + |y_1\cdots y_q| = |x_0| + |y_1\cdots y_n| + 0 \le |s|$
    \item $x_0 y_1^k x_1 \cdots y_q^k x_q = x_0 y_1^k x_1 \cdots y_n^k x_n \epsilon^k \epsilon \cdots \epsilon \epsilon^k\epsilon = \conc{s}^k \in L$ for all $k\in\N$ \qedhere
    \end{enumerate}
\end{proof}

The next two propositions show how the pumping lemma can be used to show that a language is not pure synchronized.

\begin{proposition}\label{lem:pumproof1}
    $L = \{\mathtt{a}^{n^2}\mid n\in\N\}$ is not pure synchronized.
\end{proposition}
\begin{proof}
    Assume $L$ would be pure synchronized with constants $p,q$ for Theorem~\ref{thm:pumping}. We observe that $w = \mathtt{a}^{(p+1)^2} \in L$. Now we divide $w$ into $w = x_0 y_1 x_1 \cdots y_q x_q$ with $1 \le |y_1\cdots y_q| \le p$ such that $x_0 y_1^k x_1 \cdots  y_q^k x_q \in L$ for all $k\in\N$. We compose the word $w_0 = x_0 y_1^0 x_1 \cdots y_q^0 x_q = x_0 x_1 \cdots x_q$. We observe that $|w_0| = |w| - |y_1\cdots y_q| = (p+1)^2 - |y_1\cdots y_q|$, hence $|w_0| < (p+1)^2$ and $|w_0| \ge (p+1)^2 - p = p^2+p+1 > p^2$. Since $w_0 \in L$, we have $|w_0| = |\mathtt{a}^{n^2}| = n^2$ for some $n\in\N$. However, there is no such $n\in\N$ with $p^2 < n^2 < (p+1)^2$.
\end{proof}

\begin{proposition}\label{lem:pumproof2}
    $L = \{(\mathtt{ab}^m)^n\mid n,m\in\N\}$ is not pure synchronized.
    \begin{proof}
    Assume $L$ would be pure synchronized with some constants $p,q$ for Theorem~\ref{thm:pumping}. We observe that $w = (\mathtt{ab}^{p+1})^{p+1} \in L$. The rest is left to the reader.
    \end{proof}
\end{proposition}

We will see in Section~\ref{sec:extended}, that the languages from Propositions~\ref{lem:pumproof1} and~\ref{lem:pumproof2} can be generated by a \emph{full sregex}. Hence, these languages are \emph{full synchronized}.

\subsection{Partial Concretization}

For the practical usage of concretizations, we have a problem: each concretization needs (countably) infinitely many definitions $\conc{\altern_l}_{(i)}$ and $\conc{^{\ster{l}}}_{(i)}$ for all $i\in\Id$ and $l\in\Lb$. Propositional logic with an infinite set of propositional variables $p_1,p_2,\dotsc$ has a similar problem, and the solution is using \emph{partial} valuations. We do the same for concretizations.

\begin{definition}
A partial concretization $\C = (\Id,\di, \pconcr)$ consists of partial functions for any operator with label $l\in\Lb$:
\begin{align*}
\pconc{\altern_l}&\colon \Id \pto \{ \alternleft, \alternright \}
&
\pconc{^{\ster{l}}} &\colon \Id \pto \Id^*
\end{align*}
The size of a partial concretization is measured by $\size{\C} = \size{\Id}$.
\end{definition}

Analogous to Definition~\ref{def:conc} we can construct the partial concretization $\pconc{s}$ of an sregex $s$. This is only defined, when the respective partial functions are defined.

Partial valuations in propositional logic have a small model property: to evaluate a given formula, we only need the partial valuation restricted to the propositional variables occurring in the formula. As the following example shows, such a restriction approach does not work for pure sregexes, since the number of worlds can get very large.

\begin{example}\label{ex:bigconc}
    Given is the sregex $s = (a^{\ster{1}})^{\ster{2}}$, corresponding to the regular expression $(a^*)^*$, integers $n,N\in\N$ with $N>n$, and the partial concretization $\C = (\{i_0,i_1,\dotsc,i_N\}, i_0, \pconcr)$ with:
    \begin{align*}
    \pconc{^{\ster{2}}}_{(i_0)} &= (i_1,\dotsc,i_N) \\
    \pconc{^{\ster{1}}}_{(i_1)}&= (i_1,\dotsc,i_n)\quad \text{ and }\quad \pconc{^{\ster{1}}}_{(i)}= () \quad \text{for $i=i_2,\dotsc, i_N$}
    \end{align*}
    We obtain $\pconc{(a^{\ster{1}})^{\ster{2}}} = a^n$. We get $\size{\C} = N+1$. If we take, for example, $N= 2^n$, this means that the size of the concretization is exponential to the length of the concretized word $|a^n| = n$. 
    However, we can create a second concretization $\C' = (\{i_0,\dotsc,i_n\}, i_0, \pconcr')$ with $\pconc{^{\ster{2}}}_{(i_0)}' = (i_1)$ and $\pconc{^{\ster{1}}}_{(i_1)}'=(i_1,\dotsc,i_n)$. We get $\pconc{s}' = a^n = \pconc{s}$ and $\size{\C'} = n + 1$.
    To be specific, we removed the worlds $i_{n+1},\dotsc,i_N$ and kept only $(i_1)$ in $\pconc{^{\ster{2}}}_{(i_0)}$, because $\pconc{a^{\ster{1}}}_i = \epsilon$ for all $i=i_2,\dotsc,i_N$.
\end{example}

The following theorem shows that we can always find a better partial concretization with a size polynomial in the sregex and the concretized word:

\begin{theorem}\label{thm:pconcPoly}
Given a pure sregex $s$ and a concretization $\C = (\Id, \di, \concr)$ with $w= \conc{s}$, there is a partial concretization $\C' = (\Id', \di', \pconcr)$ with $\pconc{s} = w = \conc{s}$ such that
$\size{\C'} = O(|s| \cdot |w|)$, linear in both the length of the word $|w|$ and of the sregex~$|s|$.
\end{theorem}
\begin{proof}
This is also a constructive proof.
We start with $\C$ and construct an acyclic, injective concretization as in Lemma~\ref{lem:injfull}.

Second, we restrict $\C$ to a finite partial concretization.
After the first step, we can assume that $\C$ is acyclic and injective.
When we construct $w= \conc{s}$ using Definition~\ref{def:conc}, observe that we finish in a finite number of construction steps, because each step strictly decreases the formula size of the descendants, and there are only finitely many descendants. Hence, we only need to have $\pconcr$ defined for the finite set used during the construction of $\conc{s}$ and get $\pconc{s} = \conc{s}$.
This partial concretization is finite but not necessarily small.

Step three is for each $(j_1,\dotsc,j_n) = \pconc{^{\ster{l}}}_{(i)}$ to remove all worlds $j_k$ that are not necessary, as long as the concretization of $s$ stays the same $\pconc{s} = w = \conc{s}$. It is sufficient to check that
\begin{align*}
\pconc{r^\ster{l}}_i = \pconc{r}_{j_1} \cdots \pconc{r}_{j_n} \stackrel{?}{=} \pconc{r}_{j_1} \cdots \overbrace{\quad\epsilon\quad}^{{\text{no }\pconc{r}_{j_k}}} \cdots \pconc{r}_{j_n}
\end{align*}
for all subexpression $r$ such that $\pconc{r^\ster{l}}_i$ is used when computing the concretization of $s$. Then we can redefine $\pconc{^{\ster{l}}}_{(i)} = (j_1,\dotsc,j_{k-1},j_{k+1},\dotsc,j_n)$.

Step four is for each label $l$, to set $\pconc{^\ster{l}}_{(i)}$ and $\pconc{\altern_l}_{(i)}$ to undef, when the respective values are no longer needed in the computation of $\pconc{s}$ after step three.

Now we have to check that the size of the partial concretization we constructed is bounded as claimed.
For all labels $l\in\Lb$ and indices $i\in\Lb$, there might be multiple subexpressions $r^{\ster{l}}$ such that $\pconc{r^{\ster{l}}}_i$ is used to compute $\pconc{s}$.
However, we observe with the uniqueness obtained in step one (acyclic+injective), that these subexpressions must be pairwise distinct in $s$. For example, we need $\pconc{a^{\ster{1}}}$, $\pconc{b^{\ster{1}}}$ and $\pconc{a^{\ster{1}}}$ to compute $\pconc{a^{\ster{1}}b^{\ster{1}}a^{\ster{1}}}$ and each subexpression has a pairwise distinct position in the expression.
Thus, also the respective part in the matching word $w$ might not overlap.
For each label $l\in\Lb$, there can be at most $|w|$ subexpressions $r^{\ster{l}}$ and iterations $j$, such that $\pconc{r}_{j} \neq \epsilon$ is used to compute $w = \pconc{s}$ (pigeonhole principle). Because of step three, this means that the number of iterations $j$ for each label $l\in\Lb$ is bounded by $|w|$, i.e.\ $\sum \{\mathit{len}(\pconc{^{\ster{l}}}_{(i)})\mid i\in\Id'\text{ with } \pconc{^{\ster{l}}}_{(i)} \text{ defined}\} \le |w|$.
Here, we use the length of a tuple $\mathit{len}((i_1,\dotsc,i_n)) = n$.

Now we can give an upper bound for the size of $\C'$:
\begin{align*}
        \size{\C'}\quad 
       = &\quad \size{\Id'} \\
       \le &\quad \size{\{\di'\}} + \sum \{\mathit{len}(\pconc{^\ster{l}}_{(i)})  \mid l\in\Lb,i\in \Id' \text{ with } \pconc{^\ster{l}}_{(i)} \text{ defined} \} \\
       \le &\quad 1 + \sum \{ |w| \mid l \in \Lb \text{ such that there is a subexpression } r^{\ster{l}} \text{ of }s \} \\
       \le&\quad 1+ \sum \{ |w| \mid r \text{ subexpresion of }s \} 
       \quad\le\quad 1+ |s|\cdot|w|
       \qedhere
\end{align*}
\end{proof}

From Theorem~\ref{thm:pconcPoly}, we can draw two conclusions: First, if we have some partial concretization with more than polynomial size in relation to the size of the sregex and concretized word, it is because we have numerous iterations $s_i$ of subexpression $s^{*_k}$ that simply result in the empty word $\epsilon$. Second, we can use this to prove that matching a word $w$ against an sregex $s$ is in the complexity class NP, see Section~\ref{sec:symregCompl}.

\section{Full Synchronized Regular Expressions}\label{sec:extended}

In practical regex applications,
 \emph{backreferences} are popular, which are partially handled already by sregexes with matching indices. But our approach does not work for backreferencing inside a Kleene star from outside or vice versa. We might want to generate languages like $\{(\mathtt{ab}^m)^n\mid m,n\in \N\}$, where all inner $\mathtt{ab}^*$ should match the same number of $b$ in all iterations. For a clean integration into our syntax and semantics, we introduce an \emph{uptransition} operator $(s)_\uparrow$. The following extends Definition~\ref{def:symindex}:

\begin{definition}{}\label{def:itergroups}
Given an alphabet $\Sigma$ and a set of labels $\Lb$, \emph{full sregexes over $\Sigma$} are inductively defined as all characters $s\in\Sigma$, the \emph{empty word~$\epsilon$}, the \emph{empty language~$\emptyset$}, 
concatenations $rs$,
labeled alternations $r \altern_l s$,
and labeled Kleene stars $r^{\ster{l}}$,
and \emph{uptransitions $(r)_{\uparrow}$}
for all full sregexes $r$ and $s$ and all labels $l\in\Lb$.
The uptransition has the strongest binding, so $(r)_\uparrow^\ster{l} = ((r)_\uparrow)^\ster{l}$.
\end{definition}

If we use one or multiple $\uparrow$, each uptransition overwrites the current world with the respective previous world. For example, $(\mathtt{a}(\mathtt{b}^{\ster{1}})_{\uparrow})^{\ster{3}}$ matches the language $\{(\mathtt{ab}^m)^n\mid m,n\in \N\}$. Or more compactly, with the equivalent expression $(\mathtt{ab}^{\ster{1}})_{\uparrow}^{\ster{1}}$.

When combining groups with uptransitions within star and without uptransitions outside, we can enforce the same world context within and outside. For example, $(\mathtt{ab}^{\ster{1}}) (\mathtt{ab}^{\ster{1}})_{\uparrow}^{\ster{2}}$ matches the same number of $\mathtt{b}$s outside and in each $^{\ster{2}}$ iteration.

For the semantics, concretizations need to keep track of outer iterations, too. The trick is to use full, injective concretizations only. This is not a problem, as for pure sregexes we could always construct equivalent full, injective concretizations by Lemma~\ref{lem:injfull}.

\begin{definition}{} \label{def:concgroup}
Given a full, injective concretization $\C = (\Id,\di,\concr)$ and a full sregex s, the concretization of $s$ in world $i\in\Id$ is defined by:
\begin{align*}
\conc{sr}_{i} &= \conc{s}_{i} \conc{r}_{i} \\
\conc{s \altern_l r}_{i} &= 
   \conc{s}_{i}  \conc{\altern_l}_{(i)} \conc{r}_{i}
\\
\conc{s^{\ster{l}}}_{i} &= \conc{s}_{j_1}  \conc{s}_{j_2} \cdots \conc{s}_{j_n} && \text{with $(j_1,j_2,\dotsc,j_n) = \conc{^{\ster{l}}}_{(i)}$} \\
\conc{(s)_{\uparrow}}_{i} &= \conc{s}_{j} && \text{with $(\dotsc,i,\dotsc) = \conc{^{\ster{l}}}_{(j)}$} \\
\conc{\mathtt{a}}_{i} &= \mathtt{a} && \text{for any letter $\mathtt{a}\in\Sigma$}\\
\conc{\epsilon}_i &= \epsilon\\
\conc{\emptyset}_i &= \text{\textnormal{undef}}
\end{align*}
We define the \emph{concretized word of $s$ with $\concr{}$} as $\conc{s} = \conc{s}_{\di}$.
\end{definition}

Besides requiring full and injective concretizations, the only difference to Definition~\ref{def:conc} is the addition of the uptransition $\conc{(s)_{\uparrow}}_{i} = \conc{s}_{j}$. You can see it as a stack. While the star operation $^\ster{l}$ adds a new world $j_k$ to the stack, the uptransition $_{\uparrow}$ pops one world from the stack and goes back.
For all other operations, the current world is just carried through.

\begin{proposition}
Definition~\ref{def:concgroup} is well-defined.
\end{proposition}
\begin{proof}
 For well-definedness of $\conc{(s)_\uparrow}_i$, we check that $j$ with $(\dotsc,i,\dotsc) = \conc{^\ster{l}}_{(j)}$ exists and is unique. Because the concretization is full, its graph is surjective, so there must be such a $j\in\Id$ with $(j,i)\in E$. And because the concretization is injective, there can only be one $j \in \Id$ with $(j,i)\in E$.
\end{proof}

Analogously to sregexes, we can define injective and partial concretizations, and the size of partial concretizations. Only for the pumping lemma (Theorem~\ref{thm:pumping}), we do not have a working version with iteration groups yet. Instead of Theorem~\ref{thm:pconcPoly} on partial concretizations, we now get a quadratic bound on the size of a full sregex $|s|$, see below. With that, we also get the same complexity results as in Section~\ref{sec:symregCompl}.

\begin{theorem}\label{thm:pconcPolyFull}
Given a full sregex $s$ and a full, acyclic, injective concretization $\C = (\Id, \di, \concr)$ with $w= \conc{s}$, there is a partial concretization $\C' = (\Id', \di', \pconcr)$ with $\pconc{s} = w = \conc{s}$ such that
$\size{\C'} = O(|s|^2 \cdot |w|)$, linear in $|w|$ and quadratic in $|s|$.
\end{theorem}
\begin{proof}
All four construction steps for $\C'$ are the same as in the proof of Theorem~\ref{thm:pconcPoly}.
Now we have to check that the size of the partial concretization we constructed is bounded as claimed.

For all labels $l\in\Lb$ and indices $i\in\Lb$, there might be multiple subexpressions $r^{\ster{l}}$ such that $\pconc{r^{\ster{l}}}_i$ is used to compute $\pconc{s}$.
Even though the concretization is acyclic and injective, these subexpressions may overlap in $s$. For example, we need $\pconc{b^{\ster{1}}}$ and $\pconc{(a(b^{\ster{1}})_\uparrow)^\ster{1}}$ to compute $\pconc{((a(b^{\ster{1}})_\uparrow)^\ster{1})^\ster{2}}$, both in the same worlds because of the uptransition.
Thus, for each label $l\in\Lb$, there can be at most $|s|\cdot|w|$ subexpressions $r^{\ster{l}}$ and iterations $j$, such that $\pconc{r}_{j} \neq \epsilon$ is used to compute $w = \pconc{s}$ (pigeonhole principle). Because of step three, this means that the number of iterations $j$ for each label $l\in\Lb$ is bounded by $|s|\cdot |w|$, i.e.\ $\sum \{\mathit{len}(\pconc{^{\ster{l}}}_{(i)})\mid i\in\Id'\text{ with } \pconc{^{\ster{l}}}_{(i)} \text{ defined}\} \le |s|\cdot |w|$.

Now we can give an upper bound for the size of $\C'$:
\begin{align*}
        \size{\C'}\quad 
       = &\quad \size{\Id'} \\
       \le &\quad \size{\{\di'\}} + \sum \{\mathit{len}({\pconc{^\ster{l}}_{(i)}})  \mid l\in\Lb,i\in \Id' \text{ with } \pconc{^\ster{l}}_{(i)} \text{ defined} \} \\
       \le &\quad 1 + \sum \{ |s|\cdot |w| \mid l \in \Lb \text{ such that there is a subexpression } r^{\ster{l}} \text{ of }s \} \\
       \le&\quad 1+ \sum \{ |s|\cdot|w| \mid r \text{ subexpresion of }s \} 
       \quad\le\quad 1+ |s|^2\cdot|w|
       \qedhere
\end{align*}
\end{proof}

\section{Computational Complexity of the Word Problem}\label{sec:symregCompl}

For NP-completeness, we show NP-hardness first. The next lemma is proven similar to a proof for the regular expressions with backreferences in \cite{PerlNP}.

\begin{lemma}\label{lem:NPhard}
    Given an sregex $s$ and a word $w\in\Sigma^*$, deciding $w\in\Lang(s)$ is NP-hard.
\end{lemma}
    \begin{proof}
        We show this by reduction, as 3SAT is NP-hard.
        Given is a 3CNF formula with $n$ clauses $A = (l_{1,1} \lor l_{1,2} \lor l_{1,3}) \land \cdots \land (l_{n,1}\lor l_{n,2} \lor l_{n,3})$.
        For all $i=1,\dotsc,n$, $k=1,2,3$, the literal $l_{i,k}$ is either $x_j$ or $\lnot x_j$ for some $j=1,\dotsc,n$.
        We construct the word $w = (\mathtt{ab})^n$ and the mixed expression $\alpha = (\alpha_{1,1} \altern \alpha_{1,2} \altern \alpha_{1,3}) \mathtt{b} \dotsc (\alpha_{n,1}\altern \alpha_{n,2} \altern \alpha_{n,3}) \mathtt{b}$ with
        \begin{align*}
        \alpha_{i,k} = \begin{cases} 
        (\mathtt{a} \altern_j \epsilon) &\quad \text{if }  l_{i,k} = x_j \\
        (\epsilon \altern_j \mathtt{a}) &\quad \text{if }  l_{i,k} = \lnot x_j
        \end{cases}
        \end{align*}
        for all $i=1,\dotsc,n$, $k=1,2,3$.
        By adding unique labels to the remaining $\altern$ in $\alpha$, we get an sregex $s$.
        If we find a partial concretization $\pconcr$ with $\pconc{s} = w$, we can construct a model for the formula $A$ with $x_j \iff \pconc{\altern_j}_{(0)} = \alternleft$
        for all $j = 1,\dotsc,n$.
    \end{proof}

\begin{lemma}\label{lem:inNP}
    Given an sregex $s$ and a word $w\in\Sigma^*$, deciding $w\in\Lang(s)$ is in NP.
\end{lemma}
    \begin{proof}
        We non-deterministically choose a partial concretization $\pconcr$ such that $\pconc{s} = w$ holds. We use the polynomial bound of Theorem~\ref{thm:pconcPoly} (for pure sregex) and Theorem~\ref{thm:pconcPolyFull} (for full sregex). If there is a concretization $\concr$, the algorithm finds such a partial concretization within the bound. If there is no such concretization, the algorithm fails. Hence, the algorithm runs in non-deterministic polynomial time.
    \end{proof}

\begin{theorem} \label{thm:recognize}
    Given an sregex $s$ and a word $w\in\Sigma^*$, deciding $w\in\Lang(s)$ is NP-complete.
\end{theorem}
    \begin{proof}
    The problem is NP-hard by Lemma \ref{lem:NPhard} and in NP by Lemma \ref{lem:inNP}.
    \end{proof}

\section{Closure Properties}\label{sec:close}

In this section, we show closure under union, concatenation and star for pure/full synchronized languages.
We also show that neither is closed under complement and that pure synchronized languages are not closed under intersection.
We see no point in finding a closure property for full synchronized languages or a more powerful model, because the word problem for the intersection is undecidable by Theorem~\ref{thm:undec}.

\begin{proposition}
    Pure and full synchronized languages are closed under union $L_1 \cup L_2$.
\end{proposition}
\begin{proof}
    Let $s_1,s_2$ be sregexes with $\Lang(s_1) = L_1$ and $\Lang(s_2) = L_2$.
    We can use $\Lang(s_1 \altern_l s_2) = L_1 \cup L_2$ for a fresh label $l$.
\end{proof}

\begin{proposition}
    Pure and full synchronized languages are closed under concatenation $L_1 L_2$.
\end{proposition}
\begin{proof}
    Let $s_1,s_2$ be sregexes with $\Lang(s_1) = L_1$ and $\Lang(s_2) = L_2$, such that both sregexes have a distinct set of labels.
    We can use $\Lang(s_1 s_2) = L_1 L_2$ for a fresh label $l$.
\end{proof}

For the Kleene star, the situation is not so clear anymore, because $\Lang(s)^* = \Lang(s^\ster{l})$ no longer holds.
For a full sregex $(a^\ster{1})_\uparrow$ we get $\Lang(a(b^\ster{1})_\uparrow) ^* = \{ab^m\mid m \in \N\}^* = \{ab^{m_1} \cdots ab^{m_n} \mid n\in\N, m_1,\dotsc,m_n \in \N  \}$ but 
$\Lang((a(b^\ster{1})_\uparrow)^\ster{2}) = \{ (ab^m)^n \mid m,n\in\N\}$.
We still get the closure property by removing the uptransition $_\uparrow$.
This approach does work in general, see the proof of the following proposition.

\begin{proposition}
    Pure and full synchronized languages are closed under star $L^*$.
\end{proposition}
\begin{proof}
    Let $s$ be an sregex with $\Lang(s) = L$.
    For pure sregex, we just use $\Lang(s^\ster{l}) = \Lang(s)^*$ for some fresh label $l$.

    For full sregex, we first transform $s$ into $\tilde{s}$ and also get $\Lang(s^\ster{l}) = \Lang(s)^*$ for some fresh label $l$. For each uptransition $\cdots (r)_\uparrow \cdots$, we count the outer Kleene stars $^\ster{l}$ and outer uptransitions. If the count of the Kleene stars is less or equal, then the current uptransition would go further up from the initial world $\di$ for each concretization $\C$. Hence, we remove that uptransition. To ensure that $\altern_l$ are still evaluated differently in $r$ than outside, we rename these as $\altern_{l'}$ with fresh labels $l'$. For $t^\ster{l}$, we also add fresh labels and add a new uptransition around $t$, resulting in $t_\uparrow^\ster{l'}$. We now handle the uptransitions in $t_\uparrow^\ster{l'}$ recursively.
\end{proof}

Now, we show the negative properties on complement and intersection.

\begin{proposition}
    Pure synchronized languages are not closed under intersections.
\begin{proof}
    We give pure synchronized languages $L_1$ and $L_2$ such that $L_1 \cap L_2 = \{ (a^n b)^{2m+1} \mid n,m \}$ is not pure synchronized by Theorem~\ref{thm:pumping}:
    \begin{align*}
        L_1 &= \Lang( (a^\ster{1} b a^\ster{1} b)^\ster{2} a^\ster{3} b ) &
        L_2 &= \Lang( a^\ster{1} b (a^\ster{2} b a^\ster{2} b)^\ster{3}  ) \qedhere
    \end{align*}
\end{proof}
\end{proposition}

\begin{corollary}
    Pure synchronized languages are not closed under complement, because we can express an intersection with union and complement as $L_1 \cap L_2 = (L_1^c \cup L_2^c)^c$.
\end{corollary}

The following theorem was adapted from regular expressions with backreferences, see \cite{OnExtendedRegex}.

\begin{theorem}\label{thm:undec}
    Given two pure/full sregexes $s$ and $r$, the \emph{in\-ter\-sec\-tion-non-emptiness} problem $\Lang(s) \cap \Lang(r) \stackrel{?}{=} \emptyset$ is undecidable.
\end{theorem}
\begin{proof}[Proof Sketch]
    The proof can be found in \cite[Theorem 2]{OnExtendedRegex}. The approach is the reduction to the word problem for Type-0 grammar derivations, which is undecidable.
    To show that the proof also holds for pure sregexes, we observe that all extended regular expressions used in the proof can be translated into pure sregexes as we will see in Theorem~\ref{thm:REGEX}.
\end{proof}

We conclude this section with a positive result:

\begin{theorem}\label{lemma:inPSPACE}
	Given one pure sregex $s$ and one regular expression $\alpha$, the in\-ter\-sec\-tion-non-emptiness problem $\Lang(s)\cap\Lang(\alpha) \stackrel{?}{=} \emptyset$ is in PSPACE.
\end{theorem}
For the proof, see Appendix~\ref{sec:lemma:inPSPACE}.



\section{Comparison with Related Approaches}\label{sec:expressive}

In this section, we compare the expressiveness of sregexes with language models based on backreferences and grammars.

\subsection{Backreferences}\label{sec:backreferences}

\emph{Regexes with backreferences} are found in modern programming languages such as Perl, Python and Java.
The unix utilities \texttt{grep} and \texttt{ed} support backreferences since the seventh edition of Research Unix in 1979 \cite{v7man}---here is the excerpt of the manual page \texttt{ed(1)} for the syntax and intuitive meaning:
\begin{itemize}
\item[6.] \emph{A regular expression, $x$, of form 1-8, bracketed $\langle x\rangle$ matches what $x$ matches.}
\item[7.] \emph{A $\backslash$ followed by a digit $n$ matches a copy of the string that the bracketed regular expression beginning with the $n$th $\langle$ matched.}
\end{itemize}
For example,
the expression $\langle a^*\rangle b\backslash 1$ defines the language $\{a^n b a^n \mid n \ge 0\}$.
Formal syntax and semantics are provided by several papers, e.g.~\cite{HandbookFindingPatterns,formalPracticalRegEx,SRE,OnExtendedRegex,CharacterisingRegEx,detrewbr,RegExBack}.

Note that the intuitive meaning does not cover expressions like $(y \altern \langle x\rangle)\backslash 1$ or $\langle x\rangle^*\backslash 1$, where the referenced subexpression is inside an alternation or Kleene star.
It is not clear how to resolve $\backslash 1$ when the respective $\langle x\rangle$ was never, or multiple times, matched.
There are various solutions found in different implementations and formalisms \cite{RegExBack}.
In \cite{HandbookFindingPatterns} as well as in Perl, Python and Java, the empty language~$\emptyset$ is the default when the match was never assigned, and we take the \emph{last match} when there are multiple.
Other formalizations use the empty word~$\epsilon$ as default \cite{formalPracticalRegEx,detrewbr}, which is also a good choice if we want to exclude the empty language~$\emptyset$ from the syntax and semantics of sregexes.

The following theorem shows how to transform regexes with backreferences (with $\emptyset$ as default) into sregexes:

\begin{theorem}\label{thm:REGEX}
    Every \emph{regex with backreferences} $\alpha$ can be transformed into a \emph{full sregex}.
    If $\alpha$ has no backreferences inside Kleene stars pointing to expressions outside, the transformation results in a \emph{pure sregex}.
    \begin{proof}
        If a backreference outside a star points to a subexpression $\langle\beta\rangle$ inside, we first transform $\alpha$ to an equivalent regular expression with backreferences, by replacing $(\cdots \langle \beta\rangle\cdots)^*$ with $(\epsilon \altern (\cdots \langle\beta\rangle \cdots)^*(\cdots \langle\beta\rangle\cdots))$, and adjusting backreference numbers $\backslash m$ accordingly. For each star, backreferences can only point to the last iteration, so at most one replacement is done for each star.

        We add unique labels to each alternation $\altern$ and Kleene star~$^*$ to get a synchronized expression $\tilde{\alpha}$, which still contains backreferences.
        When a backreference inside a star points to a subexpression outside the star, we add one backtransition around the backreference for each star.
        For all backreference $\backslash m$ that refer to a subexpression $\beta$ in the scope of an alternation $(\cdots \altern_l \cdots \langle \beta\rangle\cdots)$ or $(\cdots \langle\beta\rangle \cdots \altern_l \cdots)$, we add an alternation $\altern_l$ with the same label and the default option $\emptyset$ also to the backreference $\backslash m$, resulting either in $(\emptyset \altern_l \backslash m)$ or $(\backslash m \altern_l \emptyset)$.
        In case of nested alternations, we get a chain of alternations like $(\emptyset \altern_{l_1} ((\emptyset \altern_{l_3} \cdots \backslash m \cdots) \altern_{l_2} \emptyset))$.

        Finally, we replace the backreferences with the (already labeled) expressions in the parentheses they refer to.
        We also replace the brackets $\langle$ and $\rangle$ with normal brackets $($ and $)$.
        Because of the initial transformation and the added uptransitions, former references and former targets will be evaluated within the same worlds by a concretization.
    \end{proof}
\end{theorem}

\begin{example}
    The formula $\langle a^*\rangle \langle b \backslash 1 c^*\rangle \backslash 2$ satisfies the prerequisites of Theorem~\ref{thm:REGEX} for a pure sregex, because the backreference targets are not inside the scope of a Kleene star.
    First, adding labels results in $\langle a^{\ster{1}}\rangle \langle b \backslash 1 c^{\ster{2}}\rangle \backslash 2$. Then we resolve $\backslash 1$ to get $a^{\ster{1}} \langle b a^{\ster{1}} c^{\ster{2}}\rangle \backslash 2$ and $\backslash 2$ to get the pure sregex $a^{\ster{1}} (b a^{\ster{1}} c^{\ster{2}})  (b a^{\ster{1}} c^{\ster{2}})$. You can observe that the $^{\ster{1}}$ is used three times and the $^{\ster{2}}$ twice. The sregex is longer but should be easier to read, because there are no backreferences.
\end{example}

\begin{example}
The regex $\langle aaa^*\rangle \backslash 1 \backslash 1^*$ from \cite[Lemma~2]{formalPracticalRegEx} generates the language $\{ a^{mn} \mid m,n 
\ge 2\} = \{ a^m \mid \text{$m > 1$ is not a prime number} \}$.
With unique labels, we get $\langle aaa^\ster{1}\rangle \backslash 1 \backslash 1^\ster{2}$.
Now, we add uptransitions and get $\langle aaa^\ster{1}\rangle \backslash 1 (\backslash 1)_\uparrow^\ster{2}$.
We resolve the backreferences and get the full sregex $(aaa^\ster{1})(aaa^\ster{1}) (aaa^\ster{1})_\uparrow^\ster{2}$.
\end{example}

Here is one example of a backreference pointing inside an alternation:

\begin{example}
The regex $(a \altern \langle b\rangle) c \backslash 1^*$ generates the same language as $ac \altern bcb^*$.
Using the translation provided by Theorem~\ref{thm:REGEX}, we get $(a \altern_1 \langle b\rangle)c\backslash 1^\ster{2}$, then $(a \altern_1 \langle b\rangle)c(\backslash 1)_\uparrow^\ster{2}$, then  $(a \altern_1 \langle b\rangle)c(\emptyset \altern_1 \backslash 1)_\uparrow^\ster{2}$, and finally $(a \altern_1 b)c(\emptyset \altern_1 b)_\uparrow^\ster{2}$.
\end{example}

The next example shows that there are languages that can be formulated as pure sregex but not as regular expressions with backreferences.

\begin{example}
    In Example 8 of \cite{formalPracticalRegEx}, it was shown that the language $\{a^n b^n \mid n \le 0\}$ is not recognized by a regular expression with backreferences. In contrast, the pure sregex $a^{\ster{1}} b^{\ster{1}}$ recognizes it.
\end{example}

\subsection{Grammars}

Moving aside from regular languages, there are context-free and context-sensitive languages.
The language generated by the pure sregex $a^{*_1} b^{*_1} c^{*_1}$ is the non-context-free language $\{ a^n b^n c^n\mid n\in\N\}$.

The context-free grammar
    $S \to c S \text{\reflectbox{$c$}} \mid SS \mid \epsilon$
 generates the language of parentheses, PAREN$_1$, see \cite[p. 198--200]{KozenAutomata}.
Here, we fix $c$ as the left and \reflectbox{$c$} as the matching right \enquote{parentheses} for PAREN$_1$ with $n=1$.

In our intuition, the non-terminal $S$ can \enquote{call} itself with the two recursive \enquote{calls} $SS$. However, with iterations, by using the Kleene star $^{*_l}$, we cannot model such \enquote{calls}.
Of course, this is no formal proof that PAREN$_1$ is not (pure) synchronized.
We cannot use our pumping lemma (Theorem~\ref{thm:pumping}), as the language satisfies the condition.

\section{Conclusion}\label{sec:conc}

Sregexes naturally extend regular expressions, and we have shown that they have a solid theoretical foundation by using \emph{concretizations}. Similar to Kripke semantics in modal logic, word matching is done in the respective world context, with a Kleene star, we transition to other worlds. Matching is NP-complete; however, practical regexes with backreferences as used by pattern matching tools are also NP-hard.
 With $a^{*_1}b^{*_1}c^{*_1}$, we can recognize a language that is not context-free.

To show that a language is not \emph{pure} synchronized, we presented a pumping lemma. For example, the languages $\{\mathtt{a}^{n^2} \mid n\in\N\}$ and $\{(\mathtt{ab}^m)^n\mid m,n\in\N\}$ are not pure synchronized.
We also show that pure synchronized languages are closed under concatenation, star and union, but not under intersection.

As for practical tools, we compared sregex with practical regex with backreferences as defined in \cite{HandbookFindingPatterns}. Practical tools allow even more expressive regex, however, often lacking a solid theoretical foundation \cite{RegExBack}. To catch up with some practical features, it is easy to define practical sregex with more practical operators like the wildcard $\wildcard$ or quantifiers $^{n..m}$.

Another extension of regular expressions are \emph{shuffle expressions} $\alpha \shuffle \beta$ that result in letter interleavings of subexpressions $\alpha$ and $\beta$. \emph{(Letter-)synchronized shuffle expressions} \cite{Mitrana,Moreira} provide a very different form of synchronization. Exploring the combination of letter synchronization with our operator synchronization could be an interesting direction for future work.\footnote{This observation by an anonymous reviewer is very much appreciated.}

In Section~\ref{sec:extended}, we introduced \emph{full} sregex with uptransitions. With these, we can equate operators also inside different star iteration contexts. From a practical perspective, this means that we cover the expressiveness of regular expressions with backreferences. Theorem~\ref{thm:REGEX} shows this formally.

Our personal motivation was to identify program states in the dynamic logic DPDL~\cite{PDL,DPDL} and DA logic~\cite{DependentAssertionLogic}. Synchronized regular expressions are a basis for the sequent calculus in \cite[Chapter 12]{mythesis}. There, the synchronization is used for (sub)expressions that should match the same program state.

Previous work on synchronized regular expressions \cite{SRE} used a different, operational semantics.
In that paper, there are no examples for nested synchronized stars like in $(a^x b^x)^*$ and their semantics is not clear about it.
Both interpretations $\Lang((a^x b^x)^*) = \Lang((a^\ster{1}b^\ster{1})_\uparrow^\ster{2})$ and $\Lang((a^x b^x)^*) = \Lang((a^\ster{1}b^\ster{1})^\ster{2})$ are possible. Synchronized alternations $\altern_l$ are not considered in \cite{SRE}.

\subsubsection*{Acknowledgements.} I want to thank the anonymous reviewers of the ICALP 2026 for constructive criticism on my first submission under the title \enquote{Symbolic Regular Expressions}, which had less content and almost no citations. I am very thankful to the reviewers of the present conference, ICTAC 2026, for valuable feedback, which caused a second change to the current title. I also want to thank my colleagues 
for helpful discussions.

%
%
%
\bibliographystyle{splncs04}
\bibliography{DA}

\clearpage
\appendix
\renewcommand{\theHsection}{A\arabic{section}}
\section{Proof of Theorem~\ref{lemma:inPSPACE}}\label{sec:lemma:inPSPACE}

\begin{theorem*}
	Given one pure sregex $s$ and one regular expression $\alpha$, the in\-ter\-sec\-tion-non-emptiness problem $\Lang(s)\cap\Lang(\alpha) \stackrel{?}{=} \emptyset$ is in PSPACE.
\end{theorem*}
\begin{proof}
    We present a non-deterministic algorithm, running with polynomial space.
    First, we use Thompson's algorithm to convert the regular expression $\alpha$ into an NFA $\mathcal{A} = \{Q,\Sigma,\Delta,o,f\}$ with start state $o\in Q$ and end state $f\in Q$. The automaton $\mathcal{A}$ needs $O(n)$ space for the states $Q$ and transitions $\Delta \in Q\times (\Sigma \cup \{\epsilon\})\times Q$.

    For the pure sregex $s$, we use a (partial) concretization $\concr{}$. However, to meet the space limitations, for each recursive call, the algorithm will just keep the alternations $\conc{\altern_l}_{(i)}$ for a single current world $i$ in memory. This is done for all labels $l$ appearing in $s$.

    The algorithm consists of two functions: $\textsc{Match}((s_1,o_1,f_1),\dotsc,(s_k,o_k,f_k))$ matches simultanously all (sub)expressions $s_x$ with $x=1,\dotsc,k$ in the same world~$i$. While matching~$s_x$, the automaton runs from the respective states~$o_x$ to~$f_x$, which is done letter by letter in Line 17.
    The initial call is $\textsc{Match}((s,o,f))$.
    Since it is a non-deterministic algorithm, it succeeds if the final result is $\mathtt{true}$ for at least one execution path.

    The function $\textsc{MatchStar}(\mathit{count},i,(s_1^{\ster{l}},o_1,f_1),\dotsc,(s_k^{\ster{l}},o_k,f_k))$ specifically matches expressions with the same labeled star. After $\mathit{count}$ reaches $|Q|^k$, where $|Q|$ is the number of states of the automaton, all reachable state combinations have been non-deterministically tried out.
Hence, we can safely return $\mathtt{false}$.

    To show that the algorithm runs indeed in polynomial space, we observe that the recursion depth of $\textsc{Match}()$ is bounded by the number of stars in $s$. The tail-recursion in $\textsc{MatchStar}()$ adds no extra memory. The number $k$ of subexpressions in the same world is bounded by the length of the pure sregex~$s$. The $\mathit{count}$ is bounded by $n^k$, which can be represented with $k\cdot\mathit{log}(n)$ bits, hence also in polynomial space.

\vspace{1cm}
$\textsc{MatchStar} \colon \N \times (R \times Q \times Q)^* \to \{\mathtt{true},\mathtt{false}\}$
\begin{algorithmic}[1]
\Function{MatchStar}{$\mathit{count}$, $(s_1^{\ster{l}},o_1,f_1),\dotsc,(s_k^{\ster{l}},o_k,f_k)$}
    \If{$o_x=f_x$ \text{for all} $x\gets 1,\dotsc,k$}
        \textbf{return} $\mathtt{true}$
    \EndIf
    \If{$\emph{count} = |Q|^k$}
        \textbf{return} $\mathtt{false}$
    \EndIf
    \For {$x \gets 1,\dotsc,k$}
        \State non-deterministically choose $z_x \in Q$
    \EndFor
    \If{\Call{Match}{$(s_1,o_1,z_1),\dotsc,(s_k,o_k,z_k)$}$ = \mathtt{false}$}
        \textbf{return} $\mathtt{false}$
    \EndIf
    \State \textbf{return} \Call{MatchStar}{$\mathit{count}+1$, $(s_1^{\ster{l}},z_1,f_1),\dotsc,(s_k^{\ster{l}},z_k,f_k)$}
\EndFunction
\end{algorithmic}

\clearpage

$\textsc{Match} \colon (R \times Q \times Q)^* \to \{\mathtt{true},\mathtt{false}\}$
\begin{algorithmic}[1]
\Function{Match}{$(s_1,o_1,f_1),\dotsc,(s_k,o_k,f_k)$}
    \State let $i$ denote the current world
    \ForAll {labels $l$ of alternations $\altern_l$ appearing in $s_1,\dotsc,s_k$}
        \State non-deterministically choose $\pconc{\altern_l}_{(i)} \in \{\alternleft,\alternright\}$
    \EndFor
    \For {$x\gets 1,\dotsc,k$}
        \State $s_x \gets$ partially evaluate $s_x$ for all $\altern_l$ in $s_x$ with $\pconc{\altern_l}_{(i)}$
        \State \Comment{$s_x$ now consists of letters $\mathtt{a}\in\Sigma$ and unevaluated $r^{\ster{l}}$}
        \State \Comment{we say there is a \emph{position} before each $\mathtt{a}$ and $r^{\ster{l}}$}
        \State $m_x \gets $ number of positions in $s_x$
        \If{$m_x = 0$ \textbf{and} $o_x\neq f_x$}
            \textbf{return} $\mathtt{false}$
        \EndIf
        \State $z_{x,1} \gets o_x$
        \State $z_{x,m_x{+}1} \gets f_x$
        \For {$y\gets 2,\dotsc,m_x$}
            \State non-deterministically choose $z_{x,y} \in Q$
        \EndFor
        \For{$y\gets 1,\dotsc,m_x$}
            \If{$\mathtt{a} \in\Sigma$ on position $y$ of $s_x$}
                \If{NFA $\{Q,\Sigma,\Delta,z_y,z_{y+1}\}$ rejects $\mathtt{a}$}
                    \textbf{return} $\mathtt{false}$
                \EndIf
            \EndIf
        \EndFor
    \EndFor
    \ForAll {labels $l$ of stars $^\ster{l}$ appearing in $s_1,\dotsc,s_k$}
        \State $\mathit{args} \gets 0$
        \For{$x \gets 1,\dotsc,k$ \textbf{and} $y\gets 1,\dotsc,m_x$}
            \If{$r^{\ster{l}} \text{ on position }y$ of $s_x$}
                \State $\mathit{args} \gets \mathit{args}, (r^{\ster{l}},z_{x,y},z_{x,y+1})$
            \EndIf
        \EndFor
        \If{ \Call{MatchStar}{$\mathit{args}$}$ == \mathtt{false}$}
            \textbf{return} $\mathtt{false}$
        \EndIf
    \EndFor
    \State \textbf{return} $\mathtt{true}$%
\EndFunction
\qedhere
\end{algorithmic}

\end{proof}

\end{document}